\documentclass[conference]{IEEEtran}
\usepackage{cite}

\ifCLASSINFOpdf
  \usepackage[pdftex]{graphicx}
\else
  \usepackage[dvips]{graphicx}
\fi
\usepackage[cmex10]{amsmath}
\usepackage{amssymb,amsthm}
\usepackage{algorithm}
\usepackage{algpseudocode}

\usepackage{bm}
\usepackage{mdwmath}
\usepackage{mdwtab}
\usepackage{mathrsfs}   
\usepackage{amsfonts}
\usepackage{booktabs}

\usepackage[tight,footnotesize]{subfigure}

\usepackage{setspace}
\begin{document}
%
\title{Machine Learning Based Hybrid Precoding for MmWave MIMO-OFDM with Dynamic Subarray}



%
\author{\IEEEauthorblockN{Yiwei Sun\IEEEauthorrefmark{1},
Zhen Gao\IEEEauthorrefmark{2},
Hua Wang\IEEEauthorrefmark{1}, and
Di Wu\IEEEauthorrefmark{3}}
\IEEEauthorblockA{\IEEEauthorrefmark{1}School of Information and Electronics, Beijing Institute of Technology, Beijing, China}
\IEEEauthorblockA{\IEEEauthorrefmark{2}Advanced Research Institute of Multidisciplinary Science, Beijing Institute of Technology, Beijing, China}
\IEEEauthorblockA{\IEEEauthorrefmark{3}China Academy of Information and Communications Technology, Beijing, China\\
\{sunyiwei, gaozhen16, wanghua\}@bit.edu.cn}}


\maketitle

\begin{abstract}
Hybrid precoding design can be challenging
for broadband millimeter-wave (mmWave) massive MIMO due to the frequency-flat analog precoder in radio frequency (RF). Prior broadband hybrid precoding work usually focuses on fully-connected array (FCA), while seldom considers the energy-efficient partially-connected subarray (PCS) including the fixed subarray (FS) and dynamic subarray (DS). Against this background,
this paper proposes a machine learning based broadband hybrid precoding for mmWave massive MIMO with DS. Specifically, we first propose an optimal hybrid precoder based on principal component analysis (PCA) for the FS, whereby the frequency-flat RF precoder for each subarray is extracted from the principle component of the optimal frequency-selective precoders for fully-digital MIMO. Moreover, we extend the PCA-based hybrid precoding to DS, where a shared agglomerative hierarchical
clustering (AHC) algorithm developed from machine learning is proposed to group the DS for improved spectral efficiency (SE). Finally, we investigate the energy efficiency (EE) of the proposed scheme for both passive and active antennas. Simulations have confirmed that the proposed scheme outperforms conventional schemes in both SE and EE.
\end{abstract}


%
\IEEEpeerreviewmaketitle

\section{Introduction}

In the fifth generation mobile communications, the application of millimeter-wave (mmWave) is vital by the virtue of providing high data rate and large bandwidth \cite{R1,R5,G1}. However, mmWave channel suffers from a severe path loss, and the traditional fully-digital precoding with massive antennas to mitigate this issue is extremely power consuming \cite{ref_plus_5}. Therefore, hybrid precoding has been proposed to achieve the large array gains with the reduced hardware cost and power consumption \cite{liao,ref_plus_10,R4,R6}. By far, existing broadband hybrid precoding schemes usually focus on fully-connected array (FCA), while seldom consider partially-connected subarray (PCS) like fixed subarray (FS) and dynamic subarray (DS). Therefore, hybrid precoding with PCS in broadband channel is an interesting topic to explore.

Most prior work is based on narrowband mmWave channels \cite{OMP,ref_plus_9,mao}. Specifically, a compressive sensing-based
hybrid precoding has been proposed in \cite{OMP}, where the channel sparsity is ingeniously exploited to design hybrid precoding with the aid of orthogonal matching pursuit (OMP) algorithm. Moreover, a constant envelope hybrid precoding scheme is proposed, where two cost-efficient sub-connected hybrid architectures are considered to optimize the hybrid precoding under per-antenna constant envelope constraints \cite{ref_plus_9}. To improve bit-error-rate, an over-sampling codebook-based hybrid minimum sum-mean-square-error precoding is designed \cite{mao}. On the other hand, mmWave channels appear to have the frequency-selective fading, where OFDM is usually adopted to combat the time dispersion channels \cite{lim_feedback,170825,fre_flat}. Specifically, an insightful broadband hybrid precoder based on limited-feedback codebook has been proposed for FCA \cite{lim_feedback}. By exploiting the channel correlation information among different subcarriers, a broadband hybrid precoding is proposed for FS and DS \cite{170825}. Finally, \cite{fre_flat} has theoretically shown the optimality of frequency flat precoding by proving that dominant subspaces of the frequency domain channel matrices of different subcarriers are equivalent. However, this conclusion is based on the purely sparse channels with discrete angles of arrival (AoA) and angles of departure (AoD), while the explicit precoder solution is not provided.

In this paper, we propose a machine learning based hybrid precoding scheme for
mmWave MIMO-OFDM systems with DS, where a shared agglomerative hierarchical clustering (shared-AHC) algorithm is proposed for DS grouping to improve SE performance. First, we propose a PCA-based analog precoder scheme for FS by abstracting the low-dimensional signal space of frequency-flat precoder for given subarray from the high-dimensional signal space of optimal frequency-selective fully-digital precoders using PCA. Besides, the optimality of the proposed PCA-based hybrid precoder design is theoretically proven and verified by simulations. Second, we propose the shared-AHC algorithm inspired by cluster analysis in the field of machine learning for antenna grouping. By implementing shared-AHC algorithm, the SE performance of PCS can be further enhanced for effective antenna grouping adapting to the spatial features of the frequency-selective  channels. Finally, we consider the practical passive/active antennas for EE performance analysis. Simulation results confirm the better spectral efficiency (SE) and energy efficiency (EE) performance achieved by the proposed scheme than existing schemes. Meanwhile, DS has the overwhelming advantage for both active and passive antennas.

\textsl{Notations}: Following notations are used throughout this paper. $\mathbf{A}$ is a matrix, $\mathbf{a}$ is a vector, $a$ is a scalar, and $\mathcal{A}$ is a set. Conjugate transpose and transpose of $\mathbf{A}$ are $\mathbf{A}^H$ and $\mathbf{A}^T$, respectively. The $(i,j)$th entry of $\mathbf{A}$ is $[\mathbf{A}]_{i,j}$, and $[\mathbf{A}]_{i,:}$ ($[\mathbf{A}]_{:,j}$) denotes the $i$th row ($j$th column) of $\mathbf{A}$. Frobenius norm is denoted by $||\cdot||_F$. $|\mathbf{A}|$, $|\mathcal{A}|$, $|\mathbf{a}|$, and $|a|$ are the determinant of a square matrix $\mathbf{A}$, cardinality of a set $\mathcal{A}$, $\ell_2$-norm of a vector $\mathbf{a}$, and modulus of a number $a$, respectively. The $i$th largest singular value of a matrix $\mathbf{A}$ is defined as $\lambda_i(\mathbf{A})$. Additionally,
$\text{blkdiag}(\mathbf{a}_1,\cdots,\mathbf{a}_K)$ is a block diagonal matrix with $\mathbf{a}_i$ ($1\leq i\leq K$) on its diagonal blocks.

\section{System Model}
We consider an mmWave massive 3-dimensional (3D) MIMO system, where both the BS and user employ the uniform planar array (UPA), and OFDM is adopted to combat the frequency-selective fading channels. The BS is equipped with $N_t=N_t^v\times N_t^h$ antennas and $N_t^{\rm RF}\ll N_t$ chains, where $N_t^v$ and $N_t^h$ are the numbers of vertical and horizontal transmit antennas, respectively. The user is equipped with $N_r=N_r^v\times N_r^h$ antennas, where $N_r^v$ and $N_r^h$ are the numbers of vertical and horizontal receive antennas, respectively. We consider the downlink transmission, and the received symbols at the user can be written as \cite{OMP}
\begin{equation}\label{rcv_sig}
\mathbf{r}[k]=\mathbf{W}^H[k](\mathbf{H}[k]\mathbf{F}_{\rm RF}\mathbf{F}_{\rm BB}[k]\mathbf{x}[k]+\mathbf{n}[k]),
\end{equation}
where $1\leq k\leq K$ with $K$ being the number of subcarriers, $\mathbf{F}_{\rm BB}[k]\in\mathbb{C}^{N_t^{\rm RF}\times N_s}$, $\mathbf{F}_{\rm RF}\in\mathbb{C}^{N_t\times N_t^{\rm RF}}$, $\mathbf{W}[k]\in\mathbb{C}^{N_r\times N_s}$, $\mathbf{H}[k]\in\mathbb{C}^{N_r\times N_t}$, $\mathbf{x}[k]\in\mathbb{C}^{N_s\times 1}$, and $\mathbf{n}[k]\in\mathbb{C}^{N_r\times 1}$ are the digital precoder, analog precoder, fully-digital combiner, channel matrix, transmitted signal, and noise associated with the $k$th subcarrier, respectively, and $N_s$ is the number of data streams for each subcarrier. Noise $\mathbf{n}[k]$ satisfies $\mathbf{n}[k]\sim \mathcal{CN}(0,\sigma_n^2)$, and transmitted signal $\mathbf{x}[k]$ satisfies $\mathbb{E}[\mathbf{x}[k]\mathbf{x}^H\![k]]\!\!=\!\!\frac{P}{KN_s}$, where $P$ is average total transmit power.

The frequency-domain channel $\mathbf{H}[k]$ can be expressed as $\mathbf{H}[k]=\sum_{d=0}^{D-1}\mathbf{H}_d[d]e^{-j\frac{2\pi k}{K}d}$ \cite{lim_feedback}, where $D$ is the maximum delay spread of the discretized channels, and $\mathbf{H}_d[d]\in\mathbb{C}^{N_r\times N_t}$ is the delay-$d$ channel matrix. We consider the clustered channel model \cite{OMP}, where the channel is composed by $N_{\rm cl}$ clusters of multipaths with $N_{\rm ray}$ rays in each cluster. Thus the delay-$d$ channel matrix can be written as
\begin{equation}\label{H_d}
  \mathbf{H}_d[d]=\sum\nolimits_{i=1}^{N_{\rm cl}}\sum\nolimits_{l=1}^{N_{\rm ray}}p_{i,l}^d[d]\mathbf{a}_r(\phi^r_{i,l},\theta^r_{i,l})\mathbf{a}_t^H(\phi^t_{i,l},\theta^t_{i,l}),
\end{equation}
where $p_{i,l}^d[d]=\sqrt{N_tN_r/(N_{\rm cl}N_{\rm ray})}\alpha_{i,l}p(dT_s-\tau_{i,l})$ is the delay-domain channel coefficient, $\tau _{i,l}$,  $\alpha _{i,l}$, and $p(\tau)$ are the delay, the complex path gain, and the pulse shaping filter for $T_s$-spaced signaling, respectively. Thus the relationship between the frequency-domain channel coefficiency and the delay-domain channel coefficiency is $p_{i,l}[k]=\sum_{d=0}^{D-1}p_{i,l}[d]\exp(-j2\pi kd/K)$. In (\ref{H_d}), $\mathbf{a}_t(\phi^t_{i,l},\theta^t_{i,l})$ and $\mathbf{a}_r(\phi^r_{i,l},\theta^r_{i,l})$ are the steering vectors of the $l$th path in the $i$th cluster at the transmitter and receiver, respectively. In the steering vectors, $\phi^t_{i,l}$ and $\theta^t_{i,l}$ are the azimuth and elevation angles of the $l$th ray in the $i$th cluster for AoDs, and $\phi^r_{i,l}$ and $\theta^r_{i,l}$ are the azimuth and elevation angles of the $l$th ray in the $i$th cluster for AoAs. Therefore, the transmit steering vectors for the UPA at the BS can be expressed as $\mathbf{a}_t(\phi^t_{i,l},\theta^t_{i,l})=
[1\ \ \ \cdots \ \ \ e^{-j2\pi(m\frac{d_h}{\lambda}\sin(\theta ^t_{i,l})\cos(\phi ^t_{i,l})+n\frac{d_v}{\lambda}\sin(\phi ^t_{i,l}))} \ \ \
\cdots \\e^{-j2\pi((N_t^h-1)\frac{d_h}{\lambda}\sin(\theta ^t_{i,l})\cos(\phi ^t_{i,l})+(N_t^v-1)\frac{d_v}{\lambda}\sin(\phi ^t_{i,l}))}]^T/\sqrt{N_t}$ \cite{OMP}, 
where $\lambda$ is the carrier wavelength, and $d_v$ and $d_h$ are the distances between adjacent antenna elements in vertical and horizontal direction, respectively. Similarly, we can also obtain $\mathbf{a}_r(\phi ^r_{i,l},\theta^r_{i,l})$ with the same form.

The achievable rate of the mmWave MIMO heavily depends on the transmit hybrid precoder, which can be obtained by  solving the following optimization problem \cite{lim_feedback}
\begin{equation} \label{opt}
\begin{aligned}
\max\limits_{\mathbf{F}_{\rm RF},\mathbf{F}_{\rm BB}}&\!\!\sum\nolimits_{k=1}^{K}\!\!\!\!\!\log_2|\mathbf{I}\!+\!\tfrac{1}{\sigma_n^2}\mathbf{H}[k]\mathbf{F}_{\rm RF}\mathbf{F}_{\rm BB}[k]
\mathbf{F}_{\rm BB}^H[k]\mathbf{F}_{\rm RF}^H\mathbf{H}^H[k]|
\\\text{s.t. }&\mathbf{F}_{\rm RF}\in\mathcal{F}_{\rm RF},\sum\nolimits_{k=1}^K||\mathbf{F}_{\rm RF}\mathbf{F}_{\rm BB}[k]||_F^2=KN_s,
\end{aligned}
\end{equation}
where $\mathcal{F}_{\rm RF}$ is a set of feasible RF precoder satisfying constant-modulus constraint. The coupling between $\mathbf{F}_{\rm RF}$ and $\{\mathbf{F}_{\rm BB}[k]\}_{k=1}^K$ and the constant-modulus constraint of $\mathcal{F}_{\rm RF}$ lead to the challenging hybrid precoder design.

\section{PCA-Based Hybrid Precoder Design for FS}
In this section, we derive hybrid precoders for the PCS, in which only a subset of antennas are connected to each RF chain. Our goal is to design the optimal frequency-flat RF precoder from the fully-digital frequency-selective precoder.
\subsection{Digital Precoder Design}
We first design the digital precoder by fixing the RF precoder. Solving (\ref{opt}) can be difficult due to the coupling of the baseband and RF precoders\cite{lim_feedback}. Therefore, considering $\mathbf{\widetilde{F}}_{\rm BB}[k]=(\mathbf{F}_{\rm RF}^H\mathbf{F}_{\rm RF})^{\frac{1}{2}}\mathbf{F}_{\rm BB}[k]$ to be the equivalent baseband precoder,
the equivalent problem of (\ref{opt}) can be expressed as follows
\begin{equation}
\begin{aligned}\label{opt_RF}
\max\limits_{\mathbf{F}_{\rm RF},\mathbf{\widetilde{F}}_{\rm BB}}&\!\!\sum\nolimits_{k=1}^{K}\!\!\!\!\!\log_2|\mathbf{I}\!+\!\tfrac{1}{\sigma_{\rm n}^2}
\mathbf{H}[k]\mathbf{F}_{\rm RF}(\mathbf{F}_{\rm RF}^H\mathbf{F}_{\rm RF})^{-\frac{1}{2}}\mathbf{\widetilde{F}}_{\rm BB}[k]
\\&\times\mathbf{\widetilde{F}}_{\rm BB}^H[k](\mathbf{F}_{\rm RF}^H\mathbf{F}_{\rm RF})^{-\frac{1}{2}}\mathbf{F}_{\rm RF}^H\mathbf{H}^H[k]|
\\\text{s.t. }&\mathbf{F}_{\rm RF}\in\mathcal{F}_{\rm RF},\sum\nolimits_{k=1}^K||\mathbf{\widetilde{F}}_{\rm BB}[k]||_F^2=KN_s.
\end{aligned}
\end{equation}

For the optimization problem (\ref{opt_RF}), we first consider the optimal solution of $\{\mathbf{\widetilde{F}}_{\rm BB}[k]\}_{k=1}^K$. Specifically, consider the singular value decomposition (SVD) of $\mathbf{H}[k]$ associated with the $k$th subcarrier as $\mathbf{H}[k]=\mathbf{U}[k]\mathbf{\Sigma}[k]\mathbf{V}^H[k]$, and the SVD of the matrix $\mathbf{\Sigma}[k]\mathbf{V}^H[k]\mathbf{F}_{\rm RF}(\mathbf{F}_{\rm RF}^H\mathbf{F}_{\rm RF})^{-1/2}=\mathbf{\widetilde{U}}[k]\mathbf{\widetilde{\Sigma}}[k]
\mathbf{\widetilde{V}}^H[k]$. Therefore, the optimal $\mathbf{\widetilde{F}}_{\rm BB}[k]=[\mathbf{\widetilde{V}}[k]]_{:,1:N_s}\mathbf{\Lambda}[k]$, and thus the optimal baseband precoder $\mathbf{F}_{\rm BB}[k]$ can be expressed as
\begin{equation}\label{F_BB}
\begin{aligned}
\mathbf{F}_{\rm BB}[k]=&(\mathbf{F}_{\rm RF}^H\mathbf{F}_{\rm RF})^{-\frac{1}{2}}\mathbf{\widetilde{F}}_{\rm BB}[k]
\\=&(\mathbf{F}_{\rm RF}^H\mathbf{F}_{\rm RF})^{-\frac{1}{2}}[\mathbf{\widetilde{V}}[k]]_{:,1:N_s}\mathbf{\Lambda}[k],
\end{aligned}
\end{equation}
where $\mathbf{\Lambda}[k]=(\mu -N_s/[\mathbf{\widetilde{\Sigma}}[k]]_{i,i}^2)^+$ ($1\leq i\leq N_s$, $1\leq k\leq K$) is a water-filling solution matrix, in which $\mu$ satisfies $\sum_{k=1}^K\sum_{i=1}^{N_s}(\mu-N_s/[\mathbf{\widetilde{\Sigma}}[k]]_{i,i}^2)^+=KN_s$.
Then the problem reduces to obtain the optimal solution of $\mathbf{F}_{\rm RF}$ to (\ref{opt_RF}).
\subsection{PCA-Based Precoder Design}
Regarding the transmit hybrid precoder for FS, there are $N_t$ antennas and $N_t^{\rm RF}$ RF chains. For simplicity, we consider the numbers of antennas for different RF chains are identical, and the cardinality of each subset for every antenna group is $N_t^{\rm sub}=N_t/N_t^{\rm RF}$. We define the set of antenna indexes as $\{1,\cdots,N_t\}$, and $\mathcal{S}_r$ as the subset of the antennas associated with the $r$th RF chain, where $\mathcal{S}_r=\{(r-1)N_t^{\rm sub}+1,\cdots,rN_t^{\rm sub}\}$, for $1\leq r\leq N_t^{\rm RF}$. For the FS, the analog precoder $\mathbf{F}_{\rm RF}$ can be written as a block diagonal matrix $\mathbf{F}_{\rm RF}=\text{blkdiag}(\mathbf{f}_{{\rm RF},\mathcal{S}_1},\cdots,\mathbf{f}_{{\rm RF},\mathcal{S}_{N_t^{\rm RF}}})$, where $\mathbf{f}_{{\rm RF},\mathcal{S}_r}\in\mathbb{C}^{N_t^{\rm sub}\times 1}$ is the analog beamforming vector associated with the $r$th subarray for the $r$th RF chain. Defining optimal digital precoder $\mathbf{F}_{\rm opt}[k]=[\mathbf{V}[k]]_{:,1:N_s}$ for $1\leq k\leq K$, the optimal digital precoder can be expressed as
\begin{equation}\label{blk_F_opt}
\mathbf{F}_{\rm opt}^H[k]=\begin{bmatrix}
\mathbf{F}_{{\rm opt},\mathcal{S}_1}^H[k] \
\cdots \
\mathbf{F}_{{\rm opt},\mathcal{S}_{N_t^{\rm RF}}}^H[k]
\end{bmatrix},
\end{equation}
where $\mathbf{F}_{{\rm opt},\mathcal{S}_r}[k]\in\mathbb{C}^{N_t^{\rm sub}\times N_t^{\rm RF}}$. Moreover, we regard the matrix $\mathbf{F}_{\mathcal{S}_r}=\begin{bmatrix}\mathbf{F}_{{\rm opt},\mathcal{S}_r}[1]\ \mathbf{F}_{{\rm opt},\mathcal{S}_r}[2]\ \cdots\ \mathbf{F}_{{\rm opt},\mathcal{S}_r}[K]\end{bmatrix}$
consisting of the optimal precoder of all subcarriers in the $r$th subarray as the data set in the PCA problem \cite{mach_lern}. Additionally, to achieve the stable solution with low complexity for PCA, SVD is applied to the data set matrix $\mathbf{F}_{\mathcal{S}_r}$. This process is detailed in Proposition 1, where its optimality is also verified as follows.
\newtheorem{prop}{\textbf{Proposition}}
\begin{prop}
    For FS, considering $\mathbf{F}_{\mathcal{S}_r}=\begin{bmatrix}
	\mathbf{F}_{{\rm opt},\mathcal{S}_r}[1] &\cdots& \mathbf{F}_{{\rm opt},\mathcal{S}_r}[K]\end{bmatrix}$, the RF precoder $\mathbf{F}_{\rm RF}$ solving problem (\ref{opt_final_eq}) with the subarray analog/digital architecture is given by $\mathbf{F}=\text{\rm blkdiag}(\mathbf{f}_{{\rm RF},\mathcal{S}_1},\cdots,\mathbf{f}_{{\rm RF},\mathcal{S}_{N_t^{\rm RF}}})$, with $\mathbf{f}_{{\rm RF},\mathcal{S}_r}=\alpha_r\mathbf{u}_{\mathcal{S}_r}$, for $r=1,\cdots,N_t^{\rm RF}$, where $\alpha_r\in\mathbb{C}$ and $\mathbf{u}_{\mathcal{S}_r}$ is the right singular vector corresponding with the largest singular value of the matrix $\mathbf{F}_{\mathcal{S}_r}$.
\end{prop}
\begin{proof}
	Following the similar steps of the equations (12)-(14) in \cite{OMP} and defining $[\mathbf{\Sigma}[k]]_{1:N_s,1:N_s}=\mathbf{\Sigma}_1[k]$, the objective function in problem (\ref{opt}) can be approximate as
\begin{equation}
\begin{aligned}
&\sum\nolimits_{k=1}^{K}\log_2|\mathbf{I}+\tfrac{1}{\sigma_n^2}\mathbf{H}[k]\mathbf{F}_{\rm RF}\mathbf{F}_{\rm BB}[k]
\mathbf{F}_{\rm BB}^H[k]\mathbf{F}_{\rm RF}^H\mathbf{H}^H[k]|
\\\approx&\!\sum\nolimits_{k=1}^{K}\!(\log_2\!|\mathbf{I}_{N_s}\!\!+\!\!\tfrac{1}{\sigma_n^2}
\!\mathbf{\Sigma}_1^2[k]|
\!-\!(\!N_s\!\!-\!\!||\mathbf{F}_{\rm opt}^H\![k]\mathbf{F}_{\rm RF}\mathbf{F}_{\rm BB}[k]||_F^2)\!).
\end{aligned}
\end{equation}
Therefore, the optimization problem (\ref{opt}) is equivalent to the following optimization problem
\begin{equation} \label{opt_final}
\begin{aligned}
\max\limits_{\mathbf{F}_{\rm RF},\mathbf{F}_{\rm BB}}&\sum\nolimits_{k=1}^{K}||\mathbf{F}_{\rm opt}^H[k]\mathbf{F}_{\rm RF}\mathbf{F}_{\rm BB}[k]||_F^2
\\\text{s.t. }&\mathbf{F}_{\rm RF}\in\mathcal{F}_{\rm RF},\sum\nolimits_{k=1}^K||\mathbf{F}_{\rm RF}\mathbf{F}_{\rm BB}[k]||_F^2=KN_s,
\end{aligned}
\end{equation}
where $\mathcal{F}_{\rm RF}$ is a set of feasible RF precoder satisfying constant-modulus constraint.
The objective function in (\ref{opt_final}) is
\begin{equation}
\begin{aligned}
\sum\nolimits_{k=1}^{K}&\!\!\!\!\!\!
||\mathbf{F}_{\rm opt}^H\![k]\mathbf{F}_{\rm RF}\mathbf{F}_{\rm BB}[k]||_F^2
\!=\!\!\sum\nolimits_{k=1}^{K}\!\!\!\!\!\!\!\!\text{Tr}(\mathbf{F}_{\rm opt}^H\![k]\mathbf{F}_{\rm RF}\!(\mathbf{F}_{\rm RF}^H\mathbf{F}_{\rm RF}\!)\!^{-\!\frac{1}{2}}
\\&\times(\mathbf{\widetilde{F}}_{\rm BB}[k]\mathbf{\widetilde{F}}_{\rm BB}^H[k])
(\mathbf{F}_{\rm RF}^H\mathbf{F}_{\rm RF})^{-\frac{1}{2}}\mathbf{F}_{\rm RF}^H\mathbf{F}_{\rm opt}[k]).
\end{aligned}
\end{equation}
According to previous work \cite{uni_cons}, unitary constraints offer a close performance to the total power constraint while providing a relatively simple form of solution. To simplify the problem, we consider condition under unitary power constraints instead. Therefore, water-filling power allocation coefficients can be ignored. In detail, the equivalent baseband precoder $\mathbf{\widetilde{F}}_{\rm BB}[k]=[\mathbf{\widetilde{V}}[k]]_{:,1:N_s}$, which means that $\mathbf{\widetilde{F}}_{\rm BB}[k]$ is a unitary or simi-unitary matrix depending on the relationship between $N_s$ and $N_t^{\rm RF}$. When $N_s=N_t^{\rm RF}$, $\mathbf{\widetilde{F}}_{\rm BB}[k]\mathbf{\widetilde{F}}_{\rm BB}^H[k]$ is $\mathbf{I}_{N_s}$. When $N_s<N_t^{\rm RF}$, denoting the SVD of $\mathbf{\widetilde{F}}_{\rm BB}[k]=\mathbf{U}_{\rm BB}[k]\begin{bmatrix}
\mathbf{I_{N_s}} \ \mathbf{0} \end{bmatrix}^T\mathbf{V}_{\rm BB}^H[k]$, thus $\mathbf{\widetilde{F}}_{\rm BB}[k]\mathbf{\widetilde{F}}_{\rm BB}^H[k]=\mathbf{U}_{\rm BB}[k]\text{blkdiag}(\mathbf{I}_{N_s},\mathbf{0}_{N_t^{\rm RF}-N_s})\mathbf{U}_{\rm BB}^H[k]$. Therefore, the solution to the condition when $N_s=N_t^{\rm RF}$ will also suffice the condition when $N_s<N_t^{\rm RF}$. Therefore, the objective function of (\ref{opt_final}) goes down to
\begin{equation}\label{opt_final_eq}
\begin{aligned}
&\text{Tr}(\mathbf{F}_{\rm opt}^H[k]\mathbf{F}_{\rm RF}(\mathbf{F}_{\rm RF}^H\mathbf{F}_{\rm RF})^{-\frac{1}{2}}(\mathbf{F}_{\rm RF}^H\mathbf{F}_{\rm RF})^{-\frac{1}{2}}\mathbf{F}_{\rm RF}^H\mathbf{F}_{\rm opt}[k])
\\=&||\mathbf{F}_{\rm opt}^H[k]\mathbf{F}_{\rm RF}(\mathbf{F}_{\rm RF}^H\mathbf{F}_{\rm RF})^{-\frac{1}{2}}||_F^2.
\end{aligned}
\end{equation}
For simplicity, we denote $\mathbf{F}_{\rm RF}{(\mathbf{F}_{\rm RF}^H\mathbf{F}_{\rm RF})}^{-1/2}$ as $\mathbf{\bar{F}}_{\rm RF}$. Therefore, $\mathbf{\bar{F}}_{\rm RF}$ can be written into following block diagram matrix
\begin{equation}\label{F_bar}
\mathbf{\bar{F}}_{\rm RF} \!=\!\text{blkdiag}(\mathbf{f}_{{\rm RF},\mathcal{S}_1}|\mathbf{f}_{{\rm RF},\mathcal{S}_1}|^{-1}\!,\cdots,\mathbf{f}_{{\rm RF},\mathcal{S}_{N_t^{\rm RF}}}|\mathbf{f}_{{\rm RF},\mathcal{S}_{N_t^{\rm RF}}}|^{-1}).
\end{equation}
By substituting (\ref{F_bar}) and (\ref{blk_F_opt}) into (\ref{opt_final_eq}), the objective function of the optimization problem can be further expressed as
	\begin{equation}
	\begin{aligned}
	&\sum\nolimits_{k=1}^{K}||\mathbf{F}_{\rm opt}^H[k]\mathbf{\bar{F}}_{\rm RF}||_F^2
	\\=&\sum_{k=1}^{K}||\begin{bmatrix}
	\frac{\mathbf{f}_{{\rm RF},\mathcal{S}_1}\mathbf{F}_{{\rm opt},\mathcal{S}_1}^H[k]}{|\mathbf{f}_{{\rm RF},\mathcal{S}_1}|} &
	\cdots &
	\frac{\mathbf{f}_{{\rm RF},\mathcal{S}_{N_t^{\rm RF}}}\mathbf{F}_{{\rm opt},\mathcal{S}_{N_t^{\rm RF}}}^H[k]}{|\mathbf{f}_{{\rm RF},\mathcal{S}_{N_t^{\rm RF}}}|}
	\end{bmatrix}||_F^2
	\\=&\sum_{r=1}^{N_t^{\rm RF}}\frac{\mathbf{f}_{{\rm RF},\mathcal{S}_r}\mathbf{F}_{\mathcal{S}_r}^H\mathbf{F}_{\mathcal{S}_r}\mathbf{f}_{{\rm RF},\mathcal{S}_r}^H}{|\mathbf{f}_{{\rm RF},\mathcal{S}_r}|^2}.
	\end{aligned}
	\end{equation}
	Therefore, the solution to the optimization problem (\ref{opt_final_eq}) is $\max_{\mathbf{F}_{\rm RF}}\sum_{k=1}^{K}||\mathbf{\bar{F}}_{\rm RF}\mathbf{F}_{\rm opt}^H[k]||_F^2=\sum_{r=1}^{N_t^{\rm RF}}\lambda_1^2(\mathbf{F}_{\mathcal{S}_r})$. The maximum value can only be obtained when $\mathbf{f}_{{\rm RF},\mathcal{S}_r}=\alpha_r\mathbf{u}_{\mathcal{S}_r,1}$, where $\alpha_r$ is an arbitrary complex value, and $\mathbf{u}_{\mathcal{S}_r}$ is the largest singular value of the matrix $\mathbf{F}_{\mathcal{S}_r}$.
\end{proof}

Taking the constraint of RF precoder into account, we can design the RF precoder by solving
\begin{equation}\label{angle_sel}
\begin{aligned}
\mathbf{F}_{\rm RF}&=\text{blkdiag}(\mathbf{f}_{{\rm RF},\mathcal{S}_1},\cdots,\mathbf{f}_{{\rm RF},\mathcal{S}_{N_t^{\rm RF}}})
\\\text{where, }&\mathbf{f}_{{\rm RF},\mathcal{S}_r}=\arg\min_{\mathbf{x},|[\mathbf{x}]_{i,j}|=1/\sqrt{N_{\rm sub}}}||\mathbf{x}-
\mathbf{u}_{\mathcal{S}_r}||^2_F,
\\&\text{for }r=1,\cdots,N_t^{\rm RF}.
\end{aligned}
\end{equation}
With the constant-modulus constraint, the set of possible $\mathbf{f}_{{\rm RF},\mathcal{S}_r}$ is actually a hypersphere in the space of $\mathbb{C}^{N_t\times 1}$, and $\mathbf{u}_{\mathcal{S}_r}$ is a known point in the space of $\mathbb{C}^{N_t\times 1}$. Therefore, the optimization problem in (\ref{angle_sel}) is actually a distance minimization problem. Therefore, the solution is the point on this hypersphere sharing same direction of the know point $[\mathbf{f}_{{\rm RF},\mathcal{S}_r}]_{i}=\sqrt{N_{\rm sub}}e^{j\angle([\mathbf{u}_{\mathcal{S}_r}]_{i})}$.

When the quantization of phase shifters is considered, we assume the quantization bits are $Q$. Therefore, the phase shifters can only be chosen from the following quantized phase set $\mathcal{Q}=\{0,\frac{2\pi}{2^Q},\cdots,\frac{2\pi(2^Q-1)}{2^Q}\}$. Specifically, after obtaining the RF precoder $\mathbf{F}_{\rm RF}$, the quantization process can be realized by searching for the minimum Euclidean distance between $\angle([\mathbf{F}_{\rm RF}]_{i,j})$ and quantized phase from $\mathcal{Q}$.

\section{Shared-AHC Algorithm for DS Grouping}
In Section III, we have found the SE performance of FS heavily depends on $\{\mathbf{F}_{\mathcal{S}_r}\}_{r=1}^{N_t^{\rm RF}}$. This observation motivates us to optimize the antenna grouping $\{\mathcal{S}_r\}_{r=1}^{N_t^{\rm RF}}$ to further improve the SE performance when DS is considered.

The DS problem can be formulated as follows
\begin{equation}\label{pro_dy}
\begin{aligned}
&\max\limits_{\mathcal{S}_1,\cdots,\mathcal{S}_{N_t^{\rm RF}}}\sum\nolimits_{r=1}^{N_t^{\rm RF}}\lambda_1^2(\mathbf{F}_{\mathcal{S}_r})
\\&\text{s.t. }\cup_{r=1}^{N_t^{\rm RF}}\!\mathcal{S}_r\!=\!\{1,\!\cdots\!,\!N_t\},
\ \mathcal{S}_i\!\cap\!\mathcal{S}_j\!=\!\emptyset\text{ for }i\!\not=\!j,\ \mathcal{S}_r\!\not=\!\emptyset\ \forall r.
\end{aligned}
\end{equation}
This optimization problem is a combinational optimization problem, which requires an exhaustive search to reach the optimal solution. To obtain the optimal solution, the number of all possible combinations for exhaustive search can be $\frac{1}{(N_t^{\rm RF})!}\sum_{n=0}^{N_r^{\rm RF}}(-1)^{N_t^{\rm RF}-n}\binom{N_t^{\rm RF}}{n}n^{N_t}$ according to \cite{S_num}, which is a very large number. To illustrate, when $N_t=64$ and $N_t^{\rm RF}=4$, the number of all possible combinations can be up to $1.4178\times 10^{37}$.

Therefore, a low-complexity algorithm need to be develop to solve problem (\ref{pro_dy}). Specifically, we use the Minkowski $\ell_2$-norm \cite{norm} to estimate the square of the singular value of the matrix $\mathbf{F}_{\mathcal{S}_r}$ by
$\lambda_1^2(\mathbf{F}_{\mathcal{S}_r})=\lambda_1(\mathbf{R}_{\mathcal{S}_r})\approx \frac{1}{|\mathcal{S}_r|}\sum_{i=1}^{|\mathcal{S}_r|}\sum_{j=1}^{|\mathcal{S}_r|}|[\mathbf{R}_{\mathcal{S}_r}]_{i,j}|
=\frac{1}{|\mathcal{S}_r|}\sum_{i\in\mathcal{S}_r }\sum_{j\in\mathcal{S}_r}|[\mathbf{R}_F]_{i,j}|$, where $\mathbf{R}_{\mathcal{S}_r}=\mathbf{F}_{\mathcal{S}_r}\mathbf{F}_{\mathcal{S}_r}^H$ and $\mathbf{R}_F=\mathbf{F}\mathbf{F}^H$.

To reduce the complexity while achieve the good SE performance, we consider the antenna grouping from the viewpoint of clustering analysis in machine learning. To be specific, we propose a shared-AHC algorithm as listed in Algorithm \ref{alg:AHC}, which is developed from the AHC algorithm in machine
learning to group the antennas into different subarrays associated with different RF chains. Traditional AHC algorithm is a clustering algorithm that builds a cluster hierarchy from the bottom up. It starts by adding all data to multiple clusters, followed by iteratively pair-wise merging these clusters until only one cluster is left at the top of the hierarchy. The shared-AHC algorithm is different from the traditional AHC algorithm \cite{AHC} in two distinguished aspects. First, the aim of clustering in our antenna grouping problem is to build $N_t^{\rm RF}$ clusters instead of only one cluster in conventional AHC algorithm. Second, the pair-wise merging criterion in the proposed algorithm is ``shared", while the conventional AHC algorithm only considers the target cluster. To further illustrate the ``shared" mechanism, we  introduce the metric of mutual correlation $g(\mathcal{S}_n,\mathcal{S}_m)$ between the cluster $\mathcal{S}_n$ and $\mathcal{S}_m$
\begin{equation}\label{mul_cor}
g(\mathcal{S}_n,\mathcal{S}_m)=\frac{1}{|\mathcal{S}_n||\mathcal{S}_m|}
\sum_{i\in\mathcal{S}_n}\sum_{j\in\mathcal{S}_m}|[\mathbf{R}_F]_{i,j}|.
\end{equation}
In each clustering iteration, we first focus on a cluster $\mathcal{S}_n$, and find a cluster $\mathcal{S}_m$ maximizes $g(\mathcal{S}_n,\mathcal{S}_l)$ among all possible $\mathcal{S}_l$. If the cluster $\mathcal{S}_n$ also maximizes $g(\mathcal{S}_m,\mathcal{S}_l)$ among all possible $\mathcal{S}_l$, we merge $\mathcal{S}_n$ and $\mathcal{S}_m$. Otherwise, the cluster $\mathcal{S}_n$ and cluster $\mathcal{S}_m$ are not merged and algorithm goes into the next iteration. Therefore, our proposed algorithm is featured as ``shared", since two clusters mutually share the maximum correlation in the sense of (\ref{mul_cor}). This process is realized in Algorithm 1.
\begin{algorithm}[htp]
	\caption{Shared Agglomerative Hierarchical Clustering (Shared-AHC) Algorithm for DS Grouping.}
	\label{alg:AHC}
	\begin{algorithmic}[1]
		\renewcommand{\algorithmicrequire}{\textbf{Input:}}
		\renewcommand\algorithmicensure {\textbf{Output:}}
		\Require
		$\mathbf{R}_F$, number of antennas and RF chains $N_t$, $N_t^{\rm RF}$.
		\Ensure
		Grouping result $\mathcal{S}_1,\cdots,\mathcal{S}_{N_t^{\rm RF}}$.
		\State $N_{\rm sub}=N_t$, $\mathcal{S}_i=\{i\}$ for $i=1,\cdots,N_t$
		\While{$N_{\rm sub}>N_t^{\rm RF}$}
		\State $\mathcal{S}_i^0=\mathcal{S}_i$ for $i=1,\cdots,N_{\rm sub}$, $n_{\rm sub}=1$
		\For{$i=1:N_{\rm sub}$}
		\State \textbf{if} $\exists r_0\text{ s.t. }\mathcal{S}_i^0\in\mathcal{S}_{r_0}$ \textbf{then continue}
		\State \textbf{else if} $i=N_{\rm sub}$ \textbf{then} $\mathcal{S}_{n_{\rm sub}}=\mathcal{S}_i^0$
		\State \textbf{else}
		\State \qquad $j=\arg\max\limits_{l\in\{i+1,\cdots,N_{\rm sub}\}}g(\mathcal{S}_i,\mathcal{S}_l)$
        \State \qquad $i^0=\arg\max\limits_{l\in\{1,\cdots,N_{\rm sub}\}\setminus\{j\}}g(\mathcal{S}_j,\mathcal{S}_l)$
		\State \qquad \textbf{if} $i=i^0$ \textbf{then} $\mathcal{S}_{n_{\rm sub}}=\mathcal{S}_i\cup\mathcal{S}_j$
		\State \qquad \textbf{else} $\mathcal{S}_{n_{\rm sub}}=\mathcal{S}_i$
		\State \qquad\textbf{end if}
		\State \textbf{end if}
		\State $n_{\rm sub}=n_{\rm sub}+1$
		\EndFor
		\State $N_{\rm sub}^0=n_{\rm sub}-1$
		\State \textbf{if} {$N_{\rm sub}^0<N_r^{\rm RF}$} \textbf{then} $\mathcal{S}_i=\mathcal{S}_i^0$ for $i=1,\cdots,N_{\rm sub}$
        \State \qquad\textbf{break}
		\State \textbf{else} $N_{\rm sub}=N_{\rm sub}^0$
		\State \textbf{end if}
		\EndWhile
		\If{$N_{\rm sub}>N_t^{\rm RF}$}
		\State Sort $\mathcal{S}_i$ according to the ascending order of cardinality
		\For{$i=1:N_t^{\rm RF}-N_{\rm sub}$}
		\State $j=\arg\max\limits_{l=\{N_t^{\rm RF}-N_{\rm sub}+1,\cdots,N_{\rm sub}\}}g(\mathcal{S}_i,\mathcal{S}_l)$
		\State $\mathcal{S}_i=\mathcal{S}_i\cup\mathcal{S}_j$
		\EndFor
		\State Rearrange the subscript to guarantee that the order of subscripts are from 1 to $N_t^{\rm RF}$
		\EndIf
	\end{algorithmic}
\end{algorithm}
\section{Energy Efficiency Analysis}
The implementation of PCS not only reduces
the hardware complexity, but also improves the EE.
In this section, we analyze the EE of the
designs. Define the EE as $\eta =RB/P$,
where $B$ is the bandwidth of the channel, and $P$ is the total power consumption of the system.

Different connection patterns between
the phase shifters and antennas can influence
the power consumption. Because the number of phase shifters is different in different connection patterns. In this system, FCA use up to $N_tN_t^{\rm RF}$ phase shifter for each RF chain
connecting to every antennas. While the PCS use $N_t$ phase shifters.

Furthermore, the different antenna architectures should also be taken into account regarding the  total power. Specifically, we consider the hybrid MIMO system using passive antennas and active antennas as shown in Fig. 10 of \cite{ant_str}. Both of them consist electronic components such as digital-analog convertors (DAC), power amplifiers (PA), local oscillators (LO), and mixers etc. The main difference between active and passive antenna architecture lies in the number of the PAs. In passive antennas, the number of PAs is the same as that of the RF chains. While for active antennas, the number of PAs is the same as that of antennas. This difference can lead to different power consumption because the PAs are heavily power-consuming. Thus we will analyze the power consumption of the two different antenna architectures, respectively.

Given the above antenna architecture, the power consumption for FCA and PCS are respectively
$P_{\rm FCA}^p=N_tN_t^{\rm RF}P_{\rm PS}\!+\!N_t^{\rm RF}(P_{\rm DAC}\!+\!P_{\rm mix}\!+\!P_{\rm PA}\!+\!P_{\rm LO})$ and $P_{\rm PCS}^p=N_tP_{\rm PS}\!+\!N_t^{\rm RF}(P_{\rm DAC}+P_{\rm mix}\!+\!P_{\rm PA}\!+\!P_{\rm LO})$. By contrast, the power consumption for FCA and PCS with active antenna architecture are $P_{\rm FCA}^a=N_tN_t^{\rm RF}P_{\rm PS}\!+\!N_tP_{\rm PA}\!+\!N_t^{\rm RF}(P_{\rm DAC}\!+\!P_{\rm mix}\!+\!P_{\rm LO}\!)$ and
$P_{\rm PCS}^a=N_tP_{\rm PS}\!+\!N_tP_{\rm PA}\!+\!N_t^{\rm RF}(P_{\rm DAC}\!+\!P_{\rm mix}\!+\!P_{\rm LO})$.
Moreover, according to the antenna architecture for fully-digital (FD), the power consumption is
$P_{\rm FD}=N_t(P_{\rm PA}+P_{\rm DAC}+P_{\rm mix}+P_{\rm LO})$.
Additionally, the power consumption of electronic components in the three architectures are phase shifter $P_{\rm PS}=30$ mW \cite{27}, DAC $P_{\rm DAC}=200$ mW \cite{27}, mixer $P_{\rm mix}=39$ mW \cite{39}, LO $P_{\rm LO}=5$ mW \cite{27}, and PA $P_{\rm PA}=138$ mW \cite{36}.
\section{Simulations}
In this section, we investigate the SE and EE performance for the hybrid precoder design. For the channel model, we adopt Dirac delta function as the pulse shaping filter and a cyclic prefix with the length of $D=64$. The number of subcarriers is $K=512$. The transmission bandwidth is $B=500$ MHz. We consider that the path delay is uniformly distributed in $[0,DT_s]$ ($T_s=1/B$ is the symbol period). The number of the clusters is $N_{\rm cl}=8$, and azimuth/elevation AoAs and AoDs follow the uniform distribution $\mathcal{U}[-\pi/2, \pi/2]$ with angle spread of $7.5^{\circ}$. Within each cluster, there are $N_{\rm ray}=10$ rays. As for the antennas, we consider transmitter adopt $8\times8$ UPA with hybrid precoder, the receiver adopt $2\times2$ UPA with fully-digital combiner, and the distance between each adjacent antennas is half wavelength. Moreover, we consider the number of RF chains at transmitter is $N_t^{\rm RF}=4$ and the data stream is $N_s=3$. Additionally, we consider 4 types of classical FS patterns shown in Fig. \ref{FS}, where antenna elements with the same color share the same RF chain.
\begin{figure}[t]
	\centering
	\includegraphics[width=1\columnwidth, keepaspectratio]{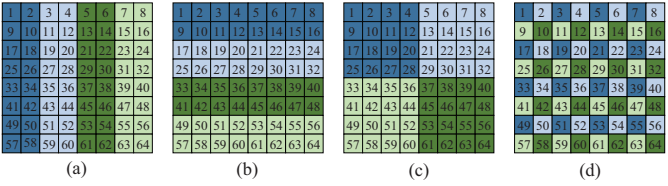}
	\caption{Four types of FS: (a) Vertical type; (b) Horizontal type; (c) Squared type; (d) Interlaced type.}\label{FS}
\end{figure}

Throughout this part, following baselines will be considered
for performance benchmarks: \textbf{Optimal fully-digital} scheme considers the fully-digital MIMO system, where the SVD-based precoder/combiner is adopted as the performance upper bound. \textbf{Simultaneous OMP (SOMP)} scheme is an extension version of the narrow-band OMP-based spatially sparse precoding in \cite{OMP}. In broadband, SOMP-based hybrid precoding scheme can simultaneously design the RF precoder/combiner for all subcarriers. \textbf{Discrete Fourier transform (DFT) codebook} scheme designs the RF precoder/combiner from the DFT codebook instead of steering vectors codebook in SOMP scheme \cite{mao}. \textbf{Covariance eigenvalue decomposition (EVD)} scheme estimates the covariance matrix of the channels using the mean of auto-correlation matrices at each subcarrier \cite{170825}. The RF precoder is designed based on the EVD of the covariance matrix of the channels.

\begin{figure}[tb]
    \centering \subfigure{\includegraphics[width=.49\columnwidth, keepaspectratio]{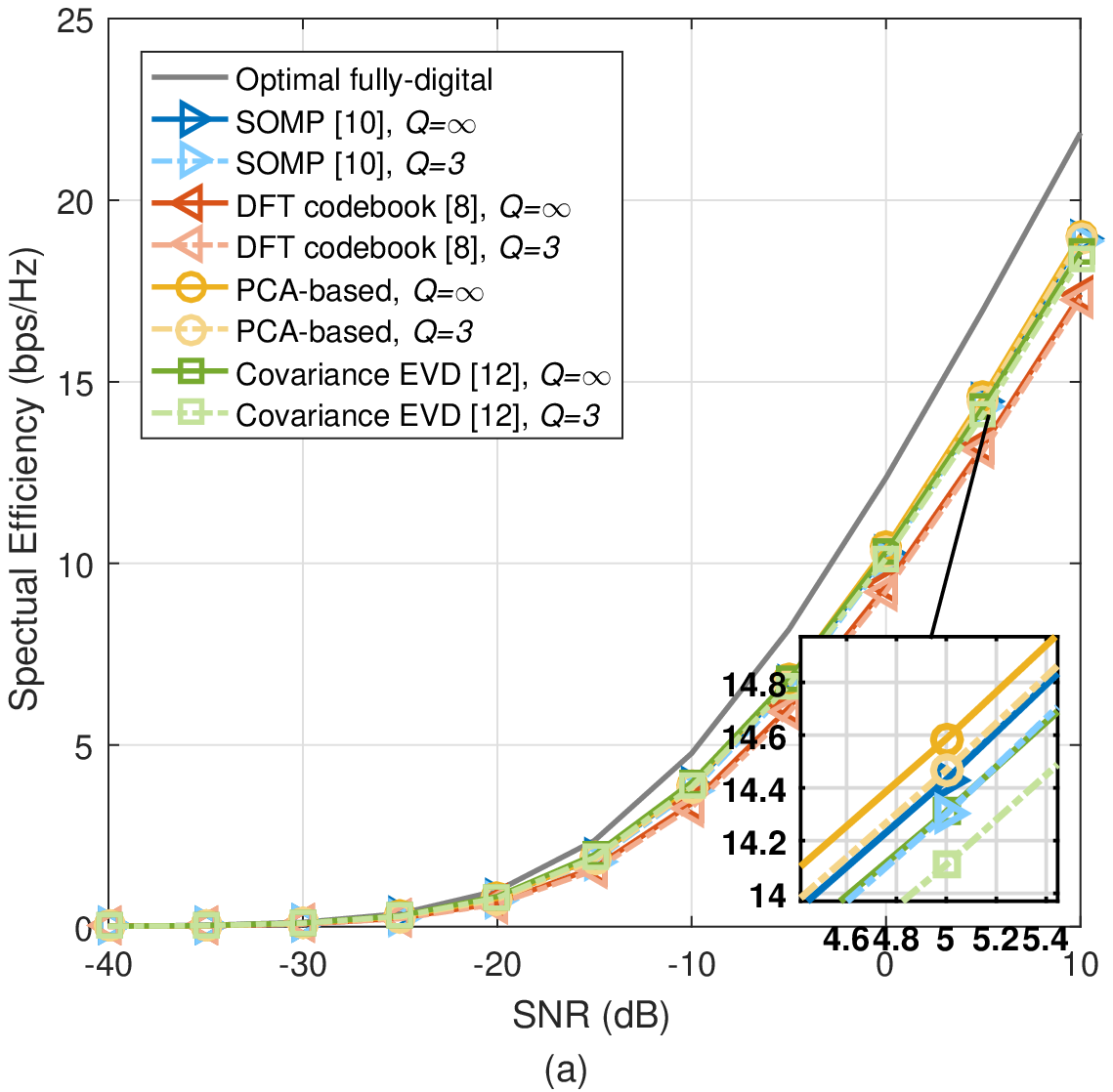}}
    \subfigure{\includegraphics[width=.49\columnwidth, keepaspectratio]{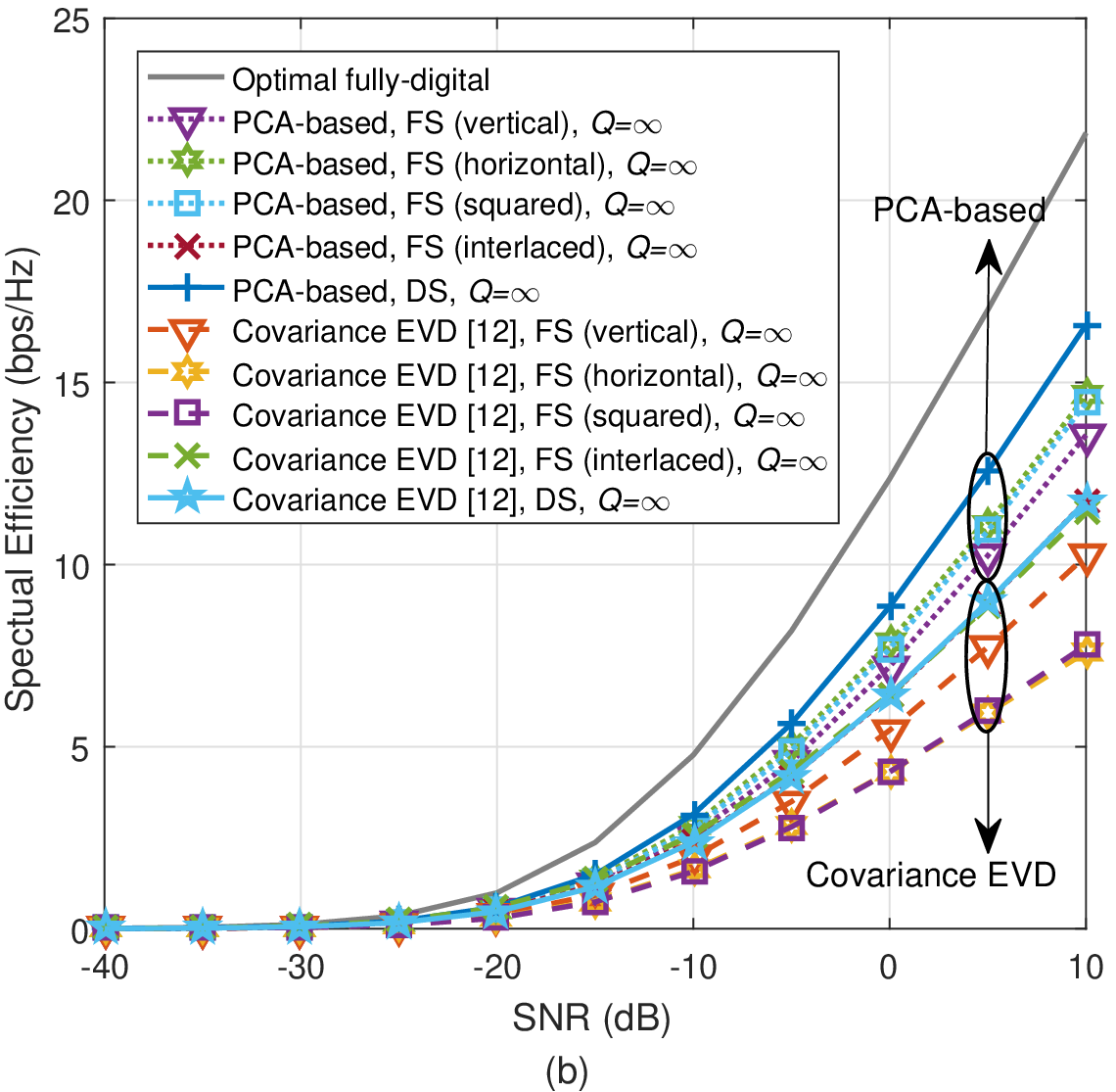}}\\
	\caption{SE performance comparison of different hybrid precoder schemes: (a) FCA with $Q=\infty$ and $Q=3$; (b) PCS with $Q=\infty$.}\label{fig_pre}
\end{figure}
In Fig. \ref{fig_pre}, we compare the SE performance of the proposed hybrid precoding scheme with the baselines, where both FCA and PCS are investigated. In Fig. \ref{fig_pre} (a), for FCA, our proposed PCA-based hybrid precoding scheme outperforms conventional DFT codebook-based hybrid precoding scheme and SOMP-based hybrid precoding scheme. Both the proposed PCA-based hybrid precoding scheme and covariance EVD-based hybrid precoding scheme have the very similar performance, and they suffer from negligible performance loss when compared to the optimal fully-digital scheme. This is because mmWave MIMO channels associated with different subcarriers share the same row space due to the common scatterers. Meanwhile, our proposed algorithm can exploit the principal components of the common row space to establish the hybrid precoder. The SOMP-based and DFT codebook-based hybrid precoding schemes work poorly, since their analog codebooks are limited to the steering vector forms. Finally, it can also be observed that the influence of quantization in phase shifters is negligible for our scheme. As for the PCS, Fig. \ref{fig_pre} (b) shows that our scheme outperforms conventional covariance EVD-based hybrid precoding scheme with different FS patterns and DS. The antenna grouping scheme in \cite{170825} considers a greedy approach, which may lead to the imbalance antenna grouping by acquiring local optimal solution. By contrast, the shared-AHC algorithm for DS grouping introduces the mutually correlation metric (\ref{mul_cor}), which can efficiently avoid this issue. Therefore, the proposed shared-AHC algorithm for DS grouping outperforms it counterpart in \cite{170825}.

\begin{figure}[tb]
    \centering \subfigure{\includegraphics[width=.49\columnwidth, keepaspectratio]{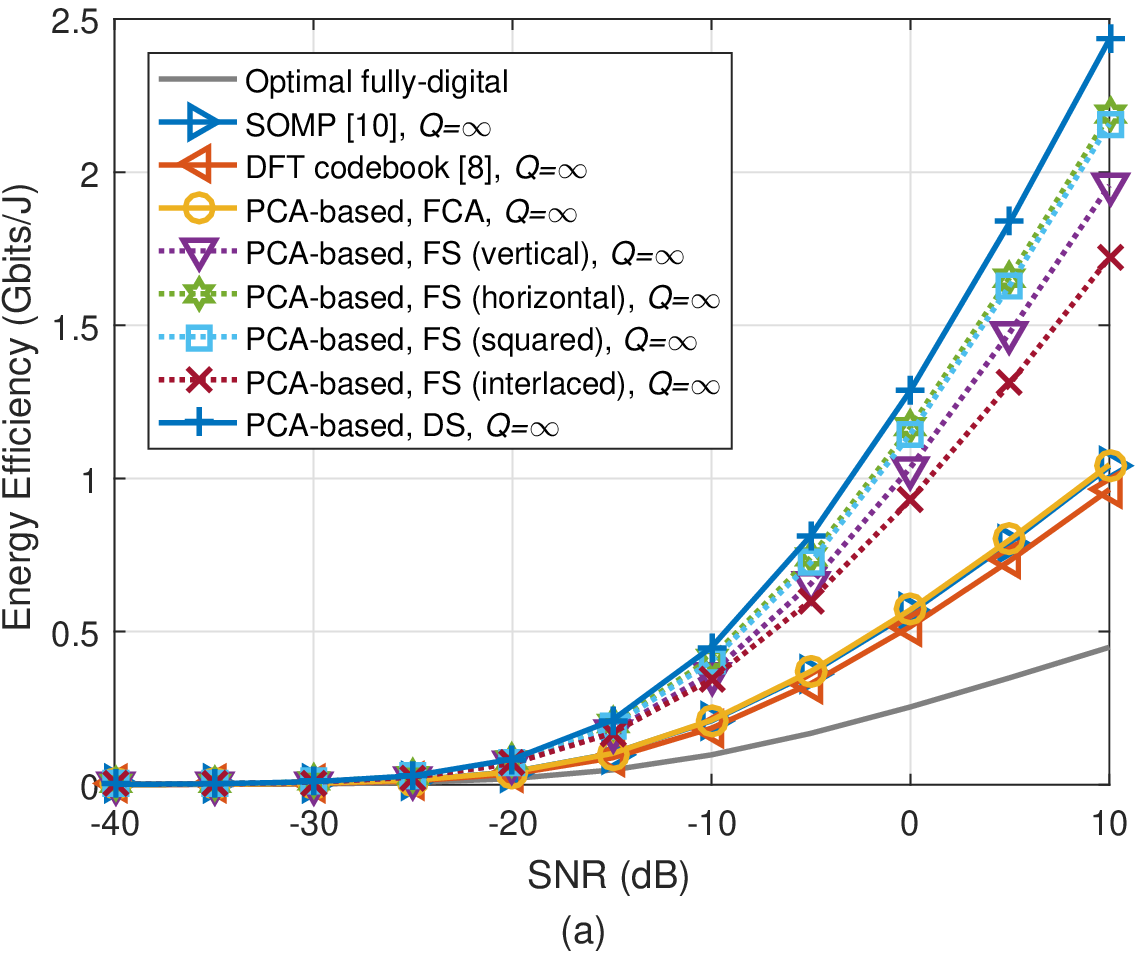}}
    \subfigure{\includegraphics[width=.49\columnwidth, keepaspectratio]{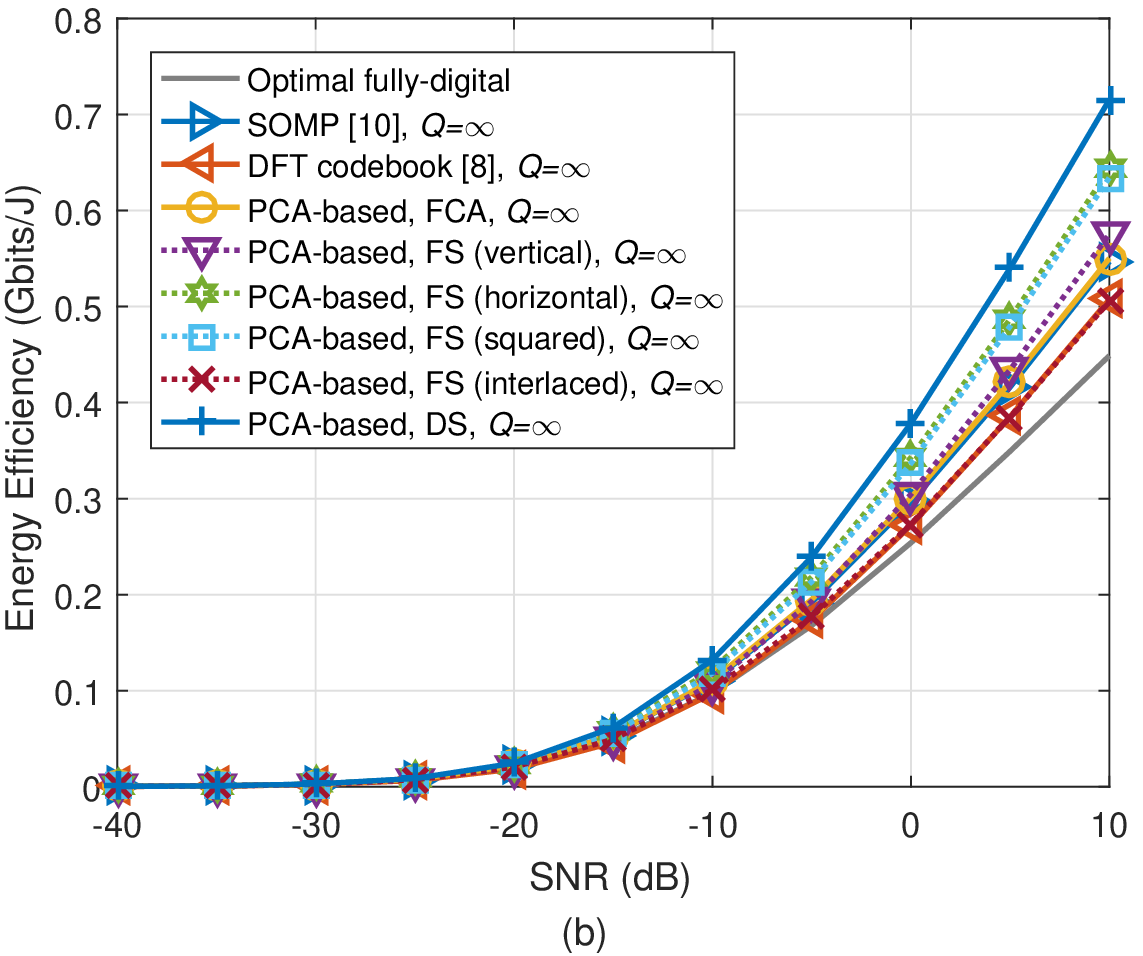}}\\
	\caption{EE performance comparison of different hybrid precoding schemes on different antenna architectures: (a) Passive antenna; (b) Active antenna with.}\label{fig_EE}
\end{figure}
In Fig. \ref{fig_EE}, we compare the EE performance of the proposed PCA-based hybrid precoding scheme and the baselines with FCA and PCS, where both passive and active antenna architectures are investigated. Note that the power values of key electronic components can refer to Section V.
In Fig. \ref{fig_EE} (a), for passive antenna architecture, the EE performance of PCS by using the proposed PCA-based hybrid precoding scheme outperforms that of FCA by using the SOMP-based and DFT codebook-based hybrid precoding schemes. The reason is that PCS adopts a much smaller number of phase shifters than FCA. Moreover, DS outperforms the other FS patterns in SE, and it consumes the same power with the other FS patterns. Therefore, DS outperforms other four types of FS patterns. It is worth mentioning that the optimal fully-digital scheme has the worst EE performance, since the numbers of power-consuming PAs, DACs, and mixers are proportional to that of antennas.
In Fig. \ref{fig_EE} (b), for active antenna architecture, the advantage of EE performance for different FS patterns by using the proposed hybrid precoding scheme over the FCA with several typical hybrid precoding schemes and optimal fully-digital precoding scheme is not considerable. This is because active antenna architecture requires the power-hungry PAs for each antenna. Meanwhile, the advantage of the reduced power consumption of FS structure is greatly weakened by its disadvantage in SE performance when compared to FCA. Finally, the EE performance of DS with the proposed hybrid precoding scheme still has the obvious advantage over the baselines and four typical types of FS with the proposed scheme. This reveals the appealing advantage of DS in practical situation when both the power consumption and SE should be well balanced.

\section{Conclusions}
This paper has proposed a hybrid precoding scheme based on machine learning for broadband mmWave MIMO
systems with DS. We first acquire the low-dimensional frequency flat precoder from the optimal frequency-selective precoders based on PCA for FS. Then, we extend the proposed PCA-based hybrid precoder design to the DS. We propose the shared-AHC algorithm inspired by cluster analysis in machine learning for antenna grouping to further improve the SE performance. Additionally, we analyze the EE performance for FCA, FS, and DS with passive and active antennas. Simulations further verify the proposed PCA-based hybrid precoding scheme has the better SE and EE performance than conventional schemes.


\section*{Acknowledgment}
This work was supported by the National Natural Science Foundation of China (Grant Nos. 61471037, 61701027, and 61201181), the Beijing Natural Science Foundation (Grant No. 4182055), Huawei Innovation Research Program (HIRP), and Youth Project of China Academy of Information and Communications Technology.





\begin{thebibliography}{1}
\bibitem{R1} Z. Xiao, P. Xia, and X. G. Xia, ``Codebook design for millimeter-wave channel estimation with hybrid precoding structure," {\em IEEE Trans. Wireless Commun.}, vol. 16, no. 1, pp. 141-153, Jan. 2017.
\bibitem{R5} Y. Sun and C. Qi, ``Weighted sum-rate maximization for analog beamforming and combining in millimeter wave massive MIMO communications," {\em IEEE Wireless Commun. Lett.}, vol. 21, no. 8, pp. 1883-1886, Oct. 2017

\bibitem{G1} Z. Gao {\em et al.}, ``MmWave massive-MIMO-based wireless backhaul for the 5G ultra-dense network," {\em IEEE Wireless Commun.}, vol. 22, no. 5, pp. 13-21, Oct. 2015.

\bibitem{ref_plus_5} Z. Gao {\em et al.}, ``Compressive sensing techniques for next-generation wireless communications,” {\em IEEE Wireless Commun.}, vol. 25, no. 3, pp. 144-153, Jun. 2018.
\bibitem{liao} A. Liao {\em et al.}, ``2D unitary ESPRIT based super-resolution channel estimation for millimeter-wave massive MIMO with hybrid precoding," {\em IEEE Access}, vol. 5, pp. 24747-24757, 2017.
\bibitem{ref_plus_10} S. He {\em et al.}, ``Codebook-based hybrid precoding for millimeter wave multiuser systems," {\em IEEE Trans. Signal Process.}, vol. 65, no. 20, pp. 5289-5304, Oct. 2017.

\bibitem{R4} A. Liu and V. K. N. Lau, ``Impact of CSI knowledge on the codebook-based hybrid beamforming in massive MIMO," {\em IEEE Trans. Signal Process.}, vol. 64, no. 24, pp. 6545-6556, Dec. 2016.
\bibitem{R6} S. He, C. Qi, Y. Wu, and Y. Huang, ``Energy-efficient transceiver design for hybrid sub-array architecture MIMO systems," {\em IEEE Access}, vol. 4, pp. 9895-9905, 2016.

\bibitem{mao} J. Mao {\em et al.}, ``Over-sampling codebook-based hybrid minimum sum-mean-square-error precoding for millimeter-wave 3D-MIMO," {\em IEEE Wireless Commun. Lett.}, vol. PP, no. PP, pp. 1-1, May 2018.

\bibitem{ref_plus_9} Y. Huang, J. Zhang, and M. Xiao, ``Constant envelope hybrid precoding for directional millimeter-wave communications," {\em IEEE J. Sel. Areas Commun.}, vol. PP, no. PP, pp. 1-1, Apr. 2018.

\bibitem{OMP}O. E. Ayach {\em et al.}, ``Spatially sparse precoding in millimeter wave MIMO systems," {\em IEEE Trans. Wireless Commun.}, vol. 13, no. 3, pp. 1499-1513, Mar. 2014.

\bibitem{lim_feedback} A. Alkhateeb and R. W. Heath Jr., ``Frequency selective hybrid precoding for limited feedback millimeter wave systems," {\em IEEE Trans. Commun.}, vol. 64, no. 5, pp. 1801-1818, May 2016.

\bibitem{170825} S. Park, A. Alkhateeb, and R. W. Heath Jr., ``Dynamic subarrays for hybrid precoding in wideband mmWave MIMO system," {\em IEEE Trans. Wireless Commun.}, vol. 16, no. 5, pp 2907-2920, May 2017.

\bibitem{fre_flat} K. Venugopal, N. G. Prelcic, and R. W. Heath Jr., ``Optimality of frequency flat precoding in frequency selective millimeter wave channels," {\em IEEE Wireless Commun. Lett.}, vol. 6, no. 3, pp. 330-333, Jun. 2017.

\bibitem{AHC} S. Zhou, Z. Xu, and F. Liu, ``Method for determining the optimal number of clusters based on agglomerative hierarchical clustering," {\em IEEE Trans. Neural Netw. Learn. Syst.}, vol. 28, no. 12, pp. 3007-3017, Dec. 2017.

\bibitem{mach_lern} C. M. Bishop, {\em Pattern Recognition and Machine Learning}. New York, NY, USA: Springer, 2006.

\bibitem{uni_cons} D. J. Love {\em et al.}, ``An overview of limited feedback in wireless communication systems," {\em IEEE J. Sel. Areas Commun.}, vol. 26, no. 8, pp. 1341-1365, Oct. 2008.

\bibitem{S_num} R. Graham, D. Knuth, and O. Patashnik, {\em Concrete Mathematics}. Reading, MA, USA: Addison-Wesley, 1988.

\bibitem{norm} H. L{\" u}tkepohl, {\rm Handbook Matrics}. Hoboken, NJ, USA: Wiley, 1996.

\bibitem{ant_str} W. Hong {\em et al.}, ``Multibeam antenna technologies for 5G wireless communications," {\em IEEE Trans. Antennas Propag.}, vol. 65, no. 12, pp. 6231-6249, Dec. 2017.

\bibitem{27} R. M{\'e}ndez-Rail {\em et al.}, ``Hybrid MIMO architectures for millimeter wave communications: phase shifters or switches?", {\em IEEE Access}, vol. 4, pp. 247-267, Jan. 2016.

\bibitem{39} M. Kraemer, D. Dragomirescu, and R. Plana, ``Design of a very low-power, low-cost 60 GHz receiver front-end implemented in 65 nm CMOS technology," {\em Int. J. Microw. Wireless Technol.}, vol. 3, pp. 131-138, Apr. 2011.

\bibitem{36} Y. Yu {\em et al.}, ``A 60 GHz phase shifter integrated with LNA and PA in 65 nm CMOS for phased array systems," {\em IEEE J. Solid-State Circuits}, vol. 45, no. 9, pp. 1697-1709, Sep. 2010.
\end{thebibliography}
%

\end{document}